\documentclass[12pt, a4paper]{article}
\usepackage{graphicx} % Required for inserting images

\usepackage{tikz}

\usepackage{amsmath,amsthm,amscd,amssymb,eucal,mathrsfs}
\usepackage{amsfonts}
\usepackage{latexsym}
\usepackage{graphicx} 
\usepackage{mathtools}
\usepackage{multicol}
\usepackage{dsfont}
\usepackage[all]{xy}
\usepackage{multirow}
\usepackage{color}
\usepackage{mathtools}
\usepackage[hidelinks]{hyperref}
\usepackage[all]{xy}
\usepackage{tcolorbox}

\newtheorem{thm}{Theorem}[section]

\newtheorem{remark}[thm]{Remark}
\newtheorem{definition}[thm]{Definition}

\newtheorem{alg}[thm]{Algorithm}

\newtheorem*{Satz*}{Satz}

\newtheorem{observation}[thm]{Observation}
\newtheorem{Lemma}[thm]{Lemma}

\newtheorem{proposition}[thm]{Proposition}

\newtheorem{Corollary}[thm]{Corollary}

\newtheorem*{algorithm}{Algorithm}

\newtheorem*{Obs}{Observation}
\newtheorem{proposal}[thm]{Proposal}
\newtheorem*{Proposal}{Proposal}
\newtheorem*{theorem}{Theorem}
\newtheorem*{corollary}{Corollary}

\newcommand{\mathset}[1]{{\left\{#1\right\}}}
\newcommand{\absolute}[1]{\left\lvert#1\right\rvert}
\newcommand{\norm}[1]{\left\|#1\right\|}

\DeclareMathOperator{\Spec}{Spec}

\DeclareMathOperator{\Vol}{Vol}

\DeclareMathOperator{\diam}{diam}

\DeclareMathOperator{\rank}{Ran}
\DeclareMathOperator{\dist}{dist}

\title{Approximating Diffusion on Finite 
Multi-Topology Systems Using Ultrametrics}
\author{Patrick Erik Bradley \\
 bradley@kit.edu \\
\'Angel Mor\'an Ledezma \\
angel.ledezma@kit.edu 
\\
Institute of Photogrammetry and Remote Sensing
\\
Karlsruhe Institute of Technology
\\
Englerstr.\ 7
\\
76131 Karlsruhe
\\
Germany
}
\date{\today}

\begin{document}

\maketitle

\begin{abstract}
Motivated by multi-topology building and city model data, first a lossless representation of multiple  $T_0$-topologies on a given finite set by a vertex-edge-weighted graph is given, and the subdominant ultrametric of the associated weighted graph distance matrix is proposed as an index structure for these data. This is applied in a heuristic parallel topological sort algorithm for edge-weighted directed acyclic graphs.
Such structured data are of interest in simulation of processes like heat flows on building or city models on distributed processors.
With this in view, the bulk of this article calculates the spectra of certain unbounded self-adjoint $p$-adic Laplacian operators on the $L^2$-spaces of a compact open subdomain of the $p$-adic number field associated with a finite graph $G$ with respect to the restricted Haar measure.
as well as to a Radon measure coming from an ultrametric on the vertices of $G$ with the help of  $p$-adic polynomial interpolation. 
In the end, error bounds are given for the solutions of the corresponding heat equations by  finite approximations of such operators.
\end{abstract}


%%%%%%%%%%5
% \ToDo{

%- finish last subsection

%- finish introduction

%+ write abstract

%- submit to Physica A?

%- or better: Journal of nonlinear mathematical physics
%}

%\ToDo{

%- question for outlook: Having a diffusion on a multi-topology, how to decompose it? Try via boundary conditions, and possibly Sobolev spaces

%- motivation for compact locally profinite spaces: it is  a framework, and \cite{XK2005} need infinite paths from root to end

%- we are proving  a locally $p$-adic generalisation of \cite[Thm.\ 10]{XK2005}

%- more motivation: finite locally ultrametric spaces approximate compact locally profinite spaces. The spectra of corresponding Laplacian matrices can be calculated via the locally $p$-adic framework!
%}

%%%%%%%%%%%%%%
\section{Introduction}

%%%%%%%%
A substantial part of data analysis deals with given finite data, viewed as nodes, and varying pairwise interaction or connectivity, viewed as edges, plus higher-order interaction or connectivity, viewed as simplices. The different graph or simplicial structures on the data help in  their understanding by exploring underlying topological and geometrical properties. Diffusion on such structures can often be used as a further tool for understanding them. The diffusion on such structures itself is also often not well understood, and this gives rise to its analytical study.
\newline

In order to facilitate a data exploration as described above, it is helpful to have an access method for the structered data, and in many cases to distribute the processing onto many different computational units. For both tasks, a hierarchical re-structuring of the data can lead to both, fast accessing and fast processing. This is one of the goals of the current research project \emph{Distributed Simulation of Processes in Buildings and City Models} in order to facilitate the simulation of processes like heat flows in city models at varying levels of resolution. In particular, if the resolution is high, then substitute models become necessary, as otherwise computing huge Laplacian matrices becomes unfeasible.
Hierarchical substitute models are natural candidates for efficient computations, as accessing trees is known to be possible in logarithmic time w.r.t.\ the number of leaf nodes. Since mathematics provides through the field $\mathds{Q}_p$ of $p$-adic numbers a natural hierarchical number system (aka ultrametric), it seemed natural to use these already 
for indexing point cloud data, i.e.\ in the case of dimension zero \cite{ScaleHilbert}.
\newline

In any case,
it becomes in the above perspective natural to resort to $p$-adic diffusion which was initiated via formulating the heat equation with the Vladimirov-Taibleson operator on $\mathds{Q}_p$ or other non-archimedean local fields \cite{Taibleson1975,VVZ1994}.
This $p$-adic Laplacian being a pseudodifferential operator necessitates the domain to be a locally compact abelian group. Another example is the group $\mathds{Z}_p$ of $p$-adic integers \cite{Kochubei2018,PW2024}. 
However, if the $p$-adic space itself does not have a group structure, then the Fourier transform is no longer available to construct a Laplacian. Its integral operator description was adapted by Z\'u\~niga-Galindo to the study of Turing patterns on networks, leading to diffusion operators on general compact open subdomains of $\mathds{Q}_p$ \cite{ZunigaNetworks}. This approach lends to the idea of asking for diffusion operators on more general $p$-adic manifolds. In the case of Mumford curves (which can be viewed as compact closed $p$-adic algebraic manifolds), this is ongoing work containing operators which allow to extract topological information from these spaces \cite{brad_HeatMumf,brad_thetaDiffusionTateCurve,HearingGenusMumf,SchottkyInvariantDiffusion}, or also from finite graphs \cite{BL_shapes_p}.  
\newline

Applications of ultrametric analysis outside of mathematics or mathematical physics include, but do not restrict to, image processing or microbiology
\cite{Angulo2019,GL2023}.
An important feature of $p$-adic or general ultrametric Laplacians is that a huge part of their spectra correspond to Haar-like wavelets of various types \cite{Kozyrev2002,GL2023}. Among $p$-adic analysts, the Kozyrev wavelets are the most well-known ones.
\newline

The goal of this contribution is to define and study $p$-adically extended Laplace operators for finite graphs, based on 
\begin{itemize}
\item weighted graph adjacency
\item weighted graph distance
\item the subdominant ultrametric of the weighted graph distance
\end{itemize}
and to investigate finite approximations of the solutions to the  heat equation corresponding to the above Laplacian operators. The envisaged input data are not merely a graph, but representations of multiple $T_0$-topologies on a finite point set (cf.\ \cite{Alexandrov1937}), as this is how data are assembled in the  research project mentioned above. The representation is given in the form of a weighted graph, and the subdominant ultrametric of the corresponding graph distance function provides  for an index structure in order to efficiently access the topologica data. The vertices of such a graph are now, as in Z\'uniga's work \cite{ZunigaNetworks}, associated with disjoint subdiscs of $\mathds{Q}_p$, and $p$-adic polynomial interpolation then allows to construct a Radon measure on this compact $p$-adic submanifold. It can be  viewed as coming from a 
regular differential $1$-form on an open piece of a Mumford curve. The natural question arises, whether and
when different datasets could be viewed as samples from open compact patches
which form a covering of a Mumford curve. 
This is accompanied by a construction  of unbounded $p$-adic diffusion operators of the three mentioned types, together with a study of their spectra, Cauchy problem and  heat kernels. 
This complements the bounded Laplacians provided by Z\'u\~niga-Galindo \cite{ZunigaNetworks}.
The last contribution consists of error estimates for the solutions of the different types of heat equations. An important concept on the way is given by \emph{ultrametric wavelets}, a generalisation of Kozyrev wavelets to arbitrary ultrametric spaces, described in \cite{XK2005}, whose corresponding trees are locally finite. Having been inspired by the Haar-like wavelets from \cite{GNC2010}, these ultrametric wavelets are presented here differently than in \cite{XK2005}   by  using characters of finite cyclic groups, while not making use of their special ultrametric distance and measure. 
\newline

The results of this study can be grouped as follows:

%%%%%%%%%%
\paragraph{1. Indexing multi-topologies.}
Motivated by the building and city model data related to the mentioned project, and which consists of points having multiple topologies, because any binary partial ordering defines a $T_0$-topologie (and vice versa) by \cite{Alexandrov1937}, an index structure can be derived by first representing the multiple topologies in a single graph, and then using the subdominant ultrametric of the graph distance. Multiple topologies on a point set arise from various different acyclic binary relations on the same set, like e.g.\ boundary and aggregation relations, or time-dependent relationships.

\begin{Obs}[\ref{losslessGraphRep}]
Multiple topologies on a finite point set have a lossless representation as a vertex-weighted graph.
\end{Obs}

Hence, a proposed index structure simply requires a weighted graph:

\begin{Proposal}[\ref{IndexStructure}]
As an index structure for a weighted graph is proposed the subdominant ultrametric of its weigted graph distance matrix.
\end{Proposal}

As an application of this index structure, a heuristic parallel topological sort is proposed:

\begin{algorithm}[\ref{ParallelTopoSort}]
Given an edge-weighted directed acyclic graph,
cluster the vertices in parallel, then sort and merge the clusters in parallel, using the subdominant ultrametric index structure.
\end{algorithm}

%%%%%%%%%%%
\paragraph{2. Laplace spectra w.r.t.\ Haar measure $dx$.}
Unbounded $p$-adic kernel functions $k_p$ depending
on the type of interaction between vertices of a finite simple connected graph are defined: adjacency, graph distance, (subdominant) ultrametric distance. The Cauchy problem for the heat equation w.r.t.\ the corresponding Laplacian operator $\mathcal{H}_\bullet$ is studied, as well as its spectrum on $L^2(Z,dx)$, where $Z\subset\mathds{Q}_p$ is the compact open subset given as the disjoint union of the $p$-adic discs defined by the vertices of a given finite simple and connected graph. In this first part, the measure used on $Z$ is the restricted Haar measure $dx$ from $\mathds{Q}_p$.

\begin{theorem}[\ref{MarkovProcess}]
There exists a probability measure $p_t(x,\cdot)$ with $t\ge0$, $x\in Z$ on the Borel $\sigma$-algebra of the compact open $Z\subset\mathds{Q}_p$ such that the Cauchy problem for the heat equation with the (unbounded) $p$-adic Laplacian operator $\mathcal{H}_\bullet$ has a unique solution of the form
\[
h(x,t)=\int_Zk_p(x,y)p_t(x,dy)
\]
In addition, $p_t(x,\cdot)$ is the transition function of a strong Markov process whose paths are c\`adl\`ag paths.
\end{theorem}

\begin{theorem}[\ref{SpectrumOfLaplacian}]
The space $L^2(Z,dx)$ has an orthonormal basis consisting of the Kozyrev wavelets supported in the discs $U_v\subset Z$ associated with vertices $v$, and the locally constant functions $Z\to\mathds{C}$ associated with the eigenvectors of a certain Laplacian matrix. These wavelets and functions are all eigenfunctions of operator $\mathcal{H}_\bullet$. In the Kozyrev case, the corresponding eigenvalue can be explicitly given. the operator $\mathcal{H}_\bullet$ is self-adjoint, negative semi-definite in case the discs associated with vertices are sufficiently small, and each eigenvalue has only finite multiplicity.
\end{theorem}

Assuming now that $\mathcal{H}_\bullet$ is also negative semi-definite, obtain:

\begin{corollary}[\ref{heatKernel}]
The heat kernel of $\mathcal{H}_\bullet$ exists and has the following form:
\[
p(t,x,y)=\sum\limits_{\lambda\in\Spec(\mathcal{H}_\bullet)}e^{\lambda t}\psi_\lambda(x)\psi_\lambda(y)
\]
for $t\ge0$ and $x,y\in Z$, where $\psi_\lambda$ is a normalised eigenfunction of $\mathcal{H}_\bullet$ with eigenvalue $\lambda\in\Spec(\mathcal{H}_\bullet)$.
\end{corollary}

%%%%%%%%%%%%5
\paragraph{3. Spectrum 
of $p$-adic operators from an ultrametric.}

Using $p$-adic polynomial interpolation, a Radon measure $\nu$ on $Z$ is constructed from an ultrametric on the vertex set of a finite graph and the $p$-adic discs representing the vertices. In general, the interpolation polynomial is non-linear of high degree, expectedly depending on the number of vertices. The corresponding Radon measure can be viewed as coming from a differential $1$-form on $Z$, giving rise to the idea that a Mumford curve might be in the background. Furtheromore, 
Inspired by \cite{GNC2010}, ultrametric wavelets generalising Kozyrev wavelets are constructed. Then it holds true that:

\begin{theorem}[\ref{SpectrumOfLaplacian2}]
The space $L^2(Z,\nu)$ has an orthonormal basis consisting of the constant function $\nu(Z)^{-1}$, and the Kozyrev wavelets supported in the discs associated with the vertices of the graph, and the ultrametric wavelets. These functions are all eigenfunctions of the operator $\mathcal{H}_\delta$. The eigenvalues corresponding to the wavelets are explicitly given in terms of the differential $1$-form and quantities derived thereof for producing the Radon measure $\nu$. The operator $\mathcal{H}_\delta$ is self-adjoint and negative-definite on $L^2(Z,\nu)$.
\end{theorem}

Here, $\delta$ is an ultrametric on the vertex set of the graph. It is a replacement for the $\bullet$ in the previous set-up, where it instances the subdominant ultrametric of a weighted graph metric.

%%%%%%%%%%%%%%
\paragraph{4. Convergence and error estimates.}
By discretising the space $Z$ via finite rooted trees resulting from pruning, approximate finite-dimensional operators for $\mathcal{H}_\bullet$ can be defined, and the errors in the solution to the Cauchy problem are estimated. 

\begin{theorem}[\ref{errorSemigroups}]
    Let $T_{\bullet}^{\ell}(t)$ and $T_{\bullet}(t)$ be the semigroups attached to the operators $\mathcal{H}_{\bullet}^{\ell}$ and $\mathcal{H}_{\bullet}$, respectively.  Let $u\in C(Z)$ and $t\in [0,\gamma]$, then
\begin{equation*}
    \begin{split}
        ||T_{\bullet}^{\ell}(t)u -T_\bullet (t)u||_{\infty}\leq  \gamma||u||_{\infty}\left(\sum_{u\neq v, [u]_{\mathfrak{n}}\equiv[v]_{\mathfrak{n}}}C_{u,v}\Vol(U_v)+\Vol(Z_{\ell}\setminus Z)\max_{x\in Z, y\in Z_{\ell}\setminus Z} |k_{\bullet}^{\ell}(x,y)|\right).
    \end{split}
\end{equation*}
\end{theorem}

\begin{theorem}[\ref{errorSemigroups2}]
    Let $\mathcal{H}_a$ and $\mathcal{H}_b$, where $a,b\in\mathset{\kappa,d_E,\delta}$. Let $T_a$ and $T_b$  the respective attached subgroups. Then the following holds 
    \[||T_{a}(t) -T_b||_{\infty} \leq 2t \left(\sum_{u\neq v, [u]_{\mathfrak{n}}\equiv[v]_{\mathfrak{n}}}\tilde{C}_{u,v}\Vol(U_v)\right).\]  
\end{theorem}

\begin{theorem}[\ref{CPerror}]
Let $u\in X\times[0,\infty)$ be a solution of the Cauchy problem for the heat equation
\begin{align*}
\frac{\partial}{\partial t}u(x,t)&-\mathcal{H}_{\bullet}u(x,t)=0
\\
u(x,0)&=u_0(x)\in X
\end{align*}
and $u_n\in X_n\times[0,\infty)$ a solution of that for
\begin{align*}
\frac{\partial}{\partial t}u_n([x]_n,t)&-\mathcal{H}_{\bullet}^{(n)} u_n([x]_n,t)=0
\\ u_n([x]_n,0)&=P_nu_0([x]_n)\in X_n
\end{align*}
Then it holds true that
\[
\lim\limits_{n\to\infty}\sup\limits_{0\le t\le\tau}\norm{E_nu_n([x]_n,t)-u(x,t)}=0
\]
for $\tau\ge 0$.
\end{theorem}

The article is structured as follows: In Section 2, finite multi-topologies are introduced as a system of finite states endowed with multple $T_0$-topologies and interpreted as a complex system with transitions w.r.t.\ different kinds of topological structures, and diffusion on such.
Section 3 defines unbounded $p$-adic graph operators on finite weighted graphs with different kinds of kernel functions representing the graph structure via adjacency, graph distance and subdominant ultrametric, defines the corresponding notion of \emph{hierarchical Parisi operator} and studies the Cauchy problems and Markov process properties. 
Section 4 estimates errors of the different kinds of operators with respect to each other via tree truncation.

%%%%%%%%%%%%%%%%%%%%
%%%%%%%%%%%%%%%%%%%%
%%%% Here is the start of the actual article
%%%%%%%%%%%%%5555
%%%%%%%%%%%%%%%5
\section{Finite multi-$T_0$-topologies}

The following setting of a complex system is considered:

\begin{itemize}
\item a finite set of states 
\item multiple directed connections between states without directed circuits
\item transitions between the states depending on the multiple connections between the states
\end{itemize}

This complex system is to be viewed as being described by transitions between multiple 
 directed acyclic graphs on a fixed finite point set.
 Each of these directed acyclic graphs can be viewed as the Hasse diagram of a finite $T_0$-space in the sense of \cite{Alexandrov1937}. So, let
 $T_1,\dots, T_N$ be a family of $T_0$-spaces on a given finite pointset $V$, and let $H_1,\dots, H_N$ be the corresponding Hasse diagrams, respectively, and let $G_i=(V,E_i)$ be the finite simple and undirected graph associated with $H_i$ for $i=1,\dots,N$.
Let
\[
E=\bigcup\limits_{i=1}^N E_i
\]
and
\[
I(e):=\left\{i\in\{1,\dots,\}\mid e\in E_i\right\}
\]
Then
\[
w\colon E\to\mathds{N},\;e\mapsto\prod\limits_{i\in I(e)}p_i^{-1}
\]
where $p_i$ is a distinct prime number associated with $i=1,\dots,N$.
Let
\[
d\colon V\to\mathds{N}^N,\;v\mapsto d_i(v)
\]
where $d_i(v)$ is the dimension of $v\in V$, viewed as a point in the $T_0$-space $T_i$, $i=1,\dots,N$. This dimension is the length of a maximal chain in $T_i$.

\begin{Lemma}\label{recoverMultiTopo}
Given the tuple $G=(V,E,w,d)$, each $T_i$ can be recovered for $i=1,\dots,N$.
\end{Lemma}

\begin{proof}
Let $e\in E$. Then, 
from the prime factorisation of $w(e)^{-1}$, the edges
$e_i\in E_i$ with $i\in I(e)$ can be recovered. By looking up $d_i(v)$ for each $v\in V$ and $i=1\dots,N$, it is now possible to find the orientation of $e_i$ via the inequality
\[
d_i(o(e_i))>d_i(t(e_i))
\]
where $o(\varepsilon)$ denotes the origin vertex, and $t(\varepsilon)$ the terminal vertex of an oriented edge $\varepsilon$ in a directed graph. This recovers each Hasse diagram $H_i$ for $i=1,\dots,N$. Since a finite $T_0$-topology is uniquely determined by its Hasse diagram, the assertion is now proven.
\end{proof}

The significance of Lemma \ref{recoverMultiTopo} is as follows:

\begin{observation}\label{losslessGraphRep}
The multiple topologies on the point set $V$ have a lossless representation as a vertex-edge-weighted graph $G=(V,E,w,d)$.
\end{observation}

%%%%%%%%%%%%%%%%%%%
%\subsection{Index structure for weighted graphs}

An index structure of a data set is a structure which allows fast access and editing of a data set. Many index structures are a tree structure defined on the data set, because trees allow traversals in logarithmic time.
For this reason:

\begin{proposal}\label{IndexStructure}
Given a finite weighted graph $G$, the tree structure given by the subdominant ultrametric of the weighted graph distance matrix of $G$ is proposed as an index structure for $G$.
\end{proposal}

Observe that the subdominant ultrametric of a metric on a finite set is unique, and can be computed by the algorithm in \cite{RTV1986}. The distance matrix for the graph $G$ can be calculated e.g.\ by Dijkstra's shortest path algorithm \cite{Dijkstra1959}.
\newline

As an application of the hierarchical structure on the vertices given by the subdominant ultrametric, a structured topological sort becomes possible. For this, assume that $H$ is a finite directed acyclic graph with weighted edges giving rise to the weighted undirected graph $G$. The directed acyclic graph structure imposes a partial ordering on $V$, denoted as $\le$.
\newline

Define $U(x)$ to be the smallest  ultrametric ball in $V$ (w.r.t.\ the subdominant ultrametric of the weighted graph distance) containing an element $x\in V$. Call it the \emph{minimal cluster of $x\in V$}.

\begin{alg}[Sorted minimal cluster]\label{clusterSort}
Do the following:
\begin{enumerate}
\item Select $x\in V$. 
\item Retrieve  $U(x)\subset V$.
\item Topologically sort $U(x)$.
\end{enumerate}
\end{alg}

Step 2.\ is done by traversing from $x$ to its parent node $v$, and then to all the child nodes of $v$ in the ultrametric tree given by the subdominant ultrametric. Step 3.\ can be done e.g.\ with Kahn's algorithm \cite{Kahn1962}.

\begin{alg}[Compare minimal clusters]\label{compareClusters}
Given $U(x)$,
do the following:
\begin{enumerate}
\item Select $y\in V\setminus U(x)$.
\item Retrieve $U(y)$.
\item Check if $x\in U(y)$. If so, then $U(x)\subseteq U(y)$. Otherwise not.
\item Check if $y\in U(x)$. If so, then $U(y)\subseteq U(y)$. Otherwise not.
\item Output one of the mutually exclusive possibilities: 
\[
U(x)\cap U(y)=\emptyset,\; U(x)\subset U(y),\;  U(y)\subset U(x),\; U(x)=U(y), 
\]
where $\subset$ denotes strict inclusion.
\end{enumerate}
\end{alg}

Steps 3.\ and 4.\ are done in a similar manner as Step 2.\ of Algorithm \ref{clusterSort}.

\begin{alg}[Merge sorted clusters]\label{mergeClusters}
Given topologically sorted $U(x)\neq U(y)$, do the following:
\begin{enumerate}
\item Retrieve the modified $T_0$-topology on $X=U(x)\cup U(y)$ induced
by the respective topological sorts of the two clusters.
\item Topologically sort $X$.\end{enumerate}
\end{alg}

Step 1.\ is given by inducing extra edges into $H$ coming from the topologically sorted clusters $U(x)$ and $U(y)$.

\begin{alg}[Parallell topological sort]\label{ParallelTopoSort}
Do the following:
\begin{enumerate}
\item With $x_1,\dots,x_n\in V$ do Algorithm \ref{clusterSort} simultaneously.
\item With any pair  $U(x_i),U(x_j)$  not yet taken, where  $i<j$, do Algorithm \ref{compareClusters}. If $U(x_i)\neq U(x_j)$ do Algorithm \ref{mergeClusters}.
\item Do Step 2.\ simultaneously for as many pairs $(i,j)$ with $i<j$, and $i,j\in\mathset{1,\dots,n}$, as the computational device can handle without problems.
\end{enumerate}
\end{alg}
This heuristic method has in general the advantage that the clusters to  be sorted and merged in parallel are not too large. However, in the worst case, all edge weights are equal, leading to the trivial metric as the corresponding subdominant ultrametric. In that case, there is only one single non-trivial cluster, and one has to resort to more traditional topological sort algorithms. But this can in principle be overcome by  modifying the graph through artificially assigning distinct weights to the edges.
\newline

The question now arises about how diffusion processes on the weighted graph $(V,E,w)$ compare with diffusions on the complete weighted graph on $V$ with weights given by either its graph metric or its subdominant ultrametric, and also the diffusion processes given by extending these operators $p$-adically, as described in the following section.

%%%%%%%%%%%%%%%%%5
\section{Unbounded $p$-adic graph operators}
\label{sec:Operators}

Assume in the following that 
Let $G=(V,E,\kappa)$ be a simple connected graph, endowed with a weight function:
\[
\kappa\colon V\times V\to [0,1]
\]
such that
\[
\kappa(x,y)=0,\;\Leftrightarrow\;(x,y)\notin E
\]
which may be a function $w$ as in the previous section.
\newline

Assume now a kernel function of a certain type
\[
k\colon V\times V\to [0,1]
\]
defining a diffusion process on $G$. The following cases will be studied:
\begin{enumerate}
\item $k$ depends on the adjacency matrix $\kappa$ of $G$:
\[
k(v,w)=\begin{cases}
\kappa(v,w)^{-\alpha},&(x,y)\in E
\\
0,&\text{otherwise}
\end{cases}
\]
with $\alpha\ge1$.
\item $k$ depends on the shortest-path-distance metric $d_E$ on $G$
\[
k(v,w)=d_E(v,w)^{-\alpha}
\]
with $\alpha\ge1$.
\item  $k$ depends on the subdominant ultrametric distance matrix $\delta$ of $d_E$:
\[
k(v,w)=\delta(v,w)^{-\alpha}
\]
with $\alpha\ge1$.
\end{enumerate}
The task is to compare the following $p$-adic diffusions having kernel function
\[
k_p\colon\mathds{Z}_p\times\mathds{Z}_p\to\mathds{R},\;(x,y)\mapsto
\begin{cases}
\absolute{x-y}_p^{-\alpha},&v(x)=v(y)
\\
k(v(x),v(y)),&v(x)\neq v(y)
\\
0,&\text{$v(x)$ or $v(y)$ do not exist}
\end{cases}
\]
with $\alpha\ge1$, and $v(x)$ the vertex containing $x\in\mathds{Z}$ (if it exists) after representing the vertices of $G$ with disjoint $p$-adic discs in $\mathds{Z}_p$ in a certain way.
\newline

Denote the corresponding $p$-adic Laplacian operator as
\[
\mathcal{H}_\bullet u(x)=\int_{\mathds{Z}_p}k_p(x,y)(u(y)-u(x))\,dy
=\int_Z k_\bullet(x,y)(u(y)-u(x))\,dy
\]
where 
\[
Z=\bigsqcup\limits_{v\in  V}v
\]
where in the disjoint union $v$ is viewed as a subset of $\mathds{Z}_p$, as explained below, and
with
\begin{align}\label{possibleBullets}
\bullet\in\mathset{\kappa,d_E,\delta}
\end{align}
and study the obvious finite approximations of $\mathcal{H}$. This operator is viewed as being attached to a graph $G_\bullet$ on $V$ having adjacency matrix $\bullet\in\mathset{\kappa,d_E,\delta}$.

\begin{definition}
The operator $\mathcal{H}_\bullet$ is called a \emph{hierarchical Parisi operator} w.r.t.\ to the measure $dx$.
\end{definition}

This is how to represent the vertices $v\in V$ as $p$-adic discs:
 $(V,\delta)$ is a finite metric space. Each vertex $v\in V$ has a nearest vertex $w\neq v$ in $V$ w.r.t.\ $\delta$. Let
 \[
r_\delta(v):=\delta(v,w)
 \]
 be the \emph{$\delta$-diameter} of $v$. Then represent each $v\in V$ with a $p$-adic disc in $\mathds{Z}_p$ of radius smaller than the $\delta$-diameter of $v$ in such a way that all these discs are disjoint.

 \begin{remark}
 The radii of the $p$-adic  discs $v(x)$ need not all be the same. But for simplicity, it will be assumed here that they are all equal.
 \end{remark}

%%%%%%%%%%%%%%%%%%%%%
\subsection{Cauchy problem and Feller process}

The operator $\mathcal{H}_\bullet$ proposed in the beginning of Section \ref{sec:Operators} is unbounded for any of the  possible choices of $\bullet$. 
For this reason, resort to the Hille-Yosida-Ray Theorem, in order to verify the Feller semigroup property.
Let 
\[
Z=\bigcup\limits_{v\in V}v\subseteq\mathds{Z}_p
\]
where each vertex $v$ is viewed as a $p$-adic disc as described in the beginning of Section \ref{sec:Operators}.

\begin{proposition}\label{FellerSemigroup}
The operator $\mathcal{H}_\bullet$ is the infinitesimal generator of a Feller semigroup $\exp\left(t\mathcal{H}_\bullet\right)$ with $t\ge0$ acting on $C(Z,\mathds{R})=C_{\mathds{R}}(Z,\norm{\cdot}_\infty)$.
\end{proposition}

\begin{proof}
Check the requirements of the Hille-Yosida-Ray Theorem:
\begin{enumerate}
\item The domain of $\mathcal{H}_\bullet$ acting on $C({Z},\norm{\cdot}_\infty)$ is dense: it contains the locally constant functions (with compact support) on $Z$, a dense subspace.
\item
The maximum positive principle: Let $\psi\in C(Z,\norm{\cdot}_\infty)$ be in the domain of $\mathcal{H}_\bullet$, and assume that 
\[
\psi(z_0)=\sup\limits_{z\in Z}\psi(z)
\]
for some $z_0\in Z$. Then
\begin{align*}
\mathcal{H}_\bullet\psi(z_0)&=
\int_{Z}k_p(z_0,z)(\psi(z)-\psi(z_0))\,dz
\\
&\le\psi(z_0)\int_{Z}k_p(z_0,z)(1-1)\,dz=0
\end{align*}
where the first inequality holds true, because the kernel function is positive.

\item The density of $\rank(\eta I-\mathcal{H})$ in $C(Z,\norm{\cdot}_\infty)$ for some $\eta>0$: The proof of \cite[Lem.\ 4.1]{ZunigaNetworks} cannot be used in this case, because the degree function is unbounded. 
So, let $h\in C(Z,\norm{\cdot}_\infty)$, and $\eta>0$, and let
\[
Z_\ell:=\mathset{x\in\mathds{Q}_p\mid p^{-\ell}\le\absolute{x}_p\le1}
\]
be an annulus with $\ell\in\mathds{N}$. The task is to find a solution of
\[
(\eta I-\mathcal{H})u=h
\]
with $u\in C(Z,\norm{\cdot}_\infty)$. Rewrite this equation as
\[
u(z)-\int_{Z}\frac{k_p(z,y)u(y)\,dy}{\eta+\deg(z)}
=\frac{h}{\eta+\deg(z)}
\]
which does not directly make sense, because $\deg(z)$ is infinite.
However, as in the proof of \cite[Lem.\ 5.1]{SchottkyInvariantDiffusion},
taking the operator
\[
T_ku(z)=\int_{Z_k}\frac{k_p(z,y)u(y)\,dy}{\eta+\deg_k(z)}
\]
with
\[
\deg_k(z)=\int_{Z_k}k(z,y)\,dy
\]
instead, does make sense, because the function $\deg_k(z)$ is finite for all $z\in Z_k$. 
It holds true that
\[
\norm{T_\ell u(z)}
\le
\frac{\deg_k(z)}{\eta+\deg_\ell(z)}\norm{u}_\infty
\]
for $\ell\in\mathds{N}$.
Hence,
\begin{align}\label{T_kBound}
\norm{T_\ell}
\le\frac{1}{\eta/\deg_k(z)+1}<1
\end{align}
%with
%\[
%\gamma_\ell=\inf\limits_{z\in Z_\ell}\deg_\ell(z)
%\]
%Taking $\eta>1-\gamma_\ell$ leads to an operator $I-T_\ell$ having an inverse on $C(Z_\ell)$.
for any $\eta>0$ and $\ell>>0$.
It follows that $I+T_\ell$ has a bounded inverse as an operator on $C(Z,\mathds{R})$. This proves the denseness of its range for $k>>0$.
\newline

Now, let $h\in\mathcal{D}(Z)$, and
$u_k,u_\ell\in C(Z,\mathds{R})$ be solutions of 
\[
(I+T_k)u_k=\frac{h}{\eta+\deg_k},\quad
(I+T_\ell)u_\ell=\frac{h}{\eta+\deg_\ell}
\]
for $k,\ell>>0$.
Then 
\[
u_k-u_\ell
= \frac{(I+T_\ell)(\eta+\deg_\ell)-(I+T_k)(\eta+\deg_k)}{(I+T_k)(I+T_\ell)(\eta+\deg_k)(\eta+\deg_\ell)}h
\]
shows that $u_k$ is a Cauchy sequence w.r.t.\ $\norm{\cdot}_\infty$.
Hence, $u_k$ converges to $u\in C(Z,\mathds{R})$ which is seen to be a solution of the limiting operator $T$ in an analogous way as in the proof of \cite[Lem.\ 3.1]{SchottkyInvariantDiffusion}.
\end{enumerate}

This now proves the assertion.
\end{proof}

\begin{remark}
The bound (\ref{T_kBound}) is tighter than Wilson's bound in the proof of \cite[Lem.\ 4.1]{ZunigaNetworks}.
\end{remark}

\begin{thm}\label{MarkovProcess}
There exists a probability measure $p_t(x,\cdot)$ with $t\ge0$, $x\in Z$ on the Borel $\sigma$-algebra of $Z$ such that the Cauchy problem for the heat equation associated with $\mathcal{H}_\bullet$ has a unique solution of the form
\[
h(x,t)=\int_{Z}k_p(x,y)p_t(x,dy)
\]
In addition, $p_t(x,\cdot)$ is the transition function of a strong Markov process whose paths are c\`adl\`ag.
\end{thm}

\begin{proof}
This follows from the well-known correspondence between Feller semi-groups and strong Markov processes with c\`adl\`ag paths.
\end{proof}

%%%%%%%%%%%%%%%%%5
\subsection{Spectrum and heat kernel}

The goal is now to make explicit the transition function of the Markov process described in Theorem \ref{MarkovProcess}.

\begin{thm}\label{SpectrumOfLaplacian}
The space $L^2(Z,dx)$ has an orthonormal basis consisting of the Kozyrev wavelets $\psi_{B,j}$, $j=1,\dots, p-1$, supported in $B\subset U_v$ for a vertex $v\in V$, and the functions associated with eigenvectors of the Laplacian matrix $L$ associated with matrix
\[
\mu_p(U_w)(k(v,w))_{v,w\in V}\in\mathds{R}^{\absolute{V}\times\absolute{V}}
\]
with $\mu_p$ the normalised Haar measure on $\mathds{Q}_p$.
These wavelets and functions are all eigenfunctions of $\mathcal{H}_\bullet$. In the Kozyrev case, the corresponding eigenvalue is
\[
\lambda_B=1-
p^{d(1+\alpha)}\left(p^{-m(1+\alpha)}+1\right)-
\sum\limits_{w\in V\setminus\mathset{v}}
k(v,w)\mu(U_w)
\]
where $B$ is a disc of radius $p^{-d}$ supported in $U_w$, a disc of radius $p^{-m}$ with integers $m\le d$.
The operator $\mathcal{H}_\bullet$ is self-adjoint, negative semi-definite if $m\ge0$, and each eigenvalue has only finite multiplicity.
\end{thm}

\begin{proof}
Assume that $\psi_{B,j}$ is a Kozyrev wavelet supported in
$B\subset U_v$. Then
\begin{align*}
\mathcal{H}_\bullet\psi_{B,j}(x)
&=\int_Zk_p(x,y)(\psi_{B,j}(y)-\psi_{B,j}(x))\,dx
\\
&=\int_{B}\absolute{x-y}^{-\alpha}
(\psi_{B,j}(y)-\psi_{B,j}(x))\,dx
\\
&+
\sum\limits_{w\in V\setminus \mathset{v}}
\int_{U_w} k(v,w)(\psi_{B,j}(y)-\psi_{B,j}(x))\,dx
\end{align*}
holds true.
The kernel function
\[
T_v(x,y):=k_p(x,y)\,\Omega(x\in U_v)\times \Omega(y\in U_v)
\]
for $x,y\in Z$
satisfies the condition 
\[
\absolute{x-y}_p=\text{const.}
\quad
\Rightarrow
\quad
T_v(x,y)=\text{const.}
\]
Hence, the conditions of \cite[Thm.\ 3]{Kozyrev2004} are satisfied, and the $p$-adic Kozyrev wavelets supported in the sets $U_v$ are eigenfunctions of the operator
\[
\mathcal{T}_v u(x)
=\int_{U_v}T_v(x,y)(u(y)-u(x))\,dy
\]
with corresponding eigenvalue
\begin{align*}
\lambda_{B,v}
&=-\int_{U_v\setminus B}
\absolute{p^{-d}n-y}_p^{-\alpha}\,dy
-p^{d(1+\alpha)}
=
-\int_{U_v\setminus B}\dist_p(B,y)^{-\alpha}
\,dy-p^{d(1+\alpha)}
\\
&=-\sum\limits_{k=-d+1}^{-m}
p^{k(1+\alpha)}(1-p^{-1})-p^{d(1+\alpha)}
=-p^{(d-m)(1+\alpha)}+1 -p^{d(1+\alpha)}
\\
&=1-p^{d(1+\alpha)}\left(p^{-m(1+\alpha)}+1\right)
\end{align*}
where 
\[
B=\mathset{z\in\mathds{Q}_p\mid
\absolute{ p^{-d}n-z}_p\le 1}
\]
is the support of the wavelet for some $n\in\mathds{Q}_p$ being the common fractional part of its residue class in $\mathds{Q}_p/\mathds{Z}_p$,
and where
\[
\dist_p(B,y)=\min\mathset{\absolute{b-y}_p\mid b\in B}
\]
It follows that
\[
\lambda_B=
\lambda_{B,v}-
\sum\limits_{w\in V\setminus\mathset{v}}
k(v,w)\mu_p(U_w)
\]
as asserted, where $\mu_p$ is the normalised Haar measure on $\mathds{Q}_p$. Observe that eigenvalue $\lambda_B$ does not depend on the choice of $j\in\mathset{1,\dots,p-1}$, but only the summand $\lambda_{B,v}$ depends on the radius $d$. In fact, this quantity is strictly increasing with $d$. Hence, the multiplicity of $\lambda_B$ is only finite.

\smallskip
Assume now that $\phi$ is a function on $Z$ coming from an eigenvector of the Laplacian matrix $L$ which belongs to a simple graph with vertex set $V$. Then, assuming that $x\in U_v$:
\begin{align*}
\mathcal{H}_\bullet\phi(x)
&=\int_Zk_p(x,y)(\phi(y)-\phi(x))\,dx
\\
&=\sum\limits_{w\in V\setminus\mathset{ v}}
k(v,w)\int_{U_w}(\phi(y)-\phi(x))\,dx
\\
&=\sum\limits_{w\in V\setminus\mathset{v}}
k(v,w) (e_w-e_v)\mu_p(U_w)
\end{align*}
where $(e_w)_{w\in V}$ is the corresponding eigenvector. It follows that
\[
\mathcal{H}_\bullet\phi(x)=\gamma e_v=\gamma\phi(x)
\]
for $x\in Z$. Hence, $\phi$ is an eigenfunction of $\mathcal{H}_\bullet$.

\smallskip
Observe now that the Kozyrev wavelets and the eigenfunctions $\phi$ are mutually orthogonal. The former span the space of $L^2$-functions on $Z$ with vanishing mean, whereas the latter span the space of functions constant on each disc $U_v$ with $v\in V$ which is clearly the orthogonal complement of the former subspace. Hence, $L^2(Z,dx)$ decomposes as an orthogonal sum  into those two subspaces, and the orthogonal basis property of these eigenfunctions of $\mathcal{H}_\bullet$ is proved. Self-adjointness of $\mathcal{H}_\bullet$ is now immediate, and negative semi-definiteness follows from the negativity of $\lambda_B$ in case $m\ge0$ together with the fact that the Laplacian matrix eigenvalues are non-positive.
This proves the theorem.
\end{proof}

Assuming now that the discs $U_v$ for vertices $v$ are sufficiently small ($m\ge0$ in Theorem \ref{SpectrumOfLaplacian}), obtain:

\begin{Corollary}\label{heatKernel}
The heat kernel of $\mathcal{H}_\bullet$ exists as a distribution and has the following form:
\[
p(t,x,y)
=\sum\limits_{\lambda\in\Spec(\mathcal{H}_\bullet)}e^{\lambda t}
\psi_\lambda(x)\psi_\lambda(y)
\]
for $t\ge0$, and $x,y\in Z$, where $\psi_\lambda$ is a normalised eigenfunction with eigenvalue $\lambda$, where by the negative semi-definiteness and the finite multiplicity properties, the sum converges.
\end{Corollary}

\begin{proof}
This is a consequence of Theorems \ref{MarkovProcess} and \ref{SpectrumOfLaplacian}. The expression defines a well-defined linear functional in the space of test functions via the assignment
\[u(x)\mapsto \int_{Z} u(y)p(x,y,t)dy.\]
\end{proof}

%%%%%%%%%%%%%%%%%%%%
\subsection{A $p$-adic measure from a finite ultrametric}\label{RadonMeasure}

Let $\delta$ be an ultrametric on the finite set $V$. 
The goal here is to define a $p$-adic measure compatible with $\delta$ in some sense.
\newline

The first step is to
fix an embedding
\begin{align}\label{treeEmbedding}
T_\delta\to\mathscr{T}_p
\end{align}
which takes the root to the vertex corresponding to $\mathds{Z}_p$, and all other nodes to subdiscs of $\mathds{Z}_p$ such that $T_\delta$ is isomorphic to a subtree of $\mathscr{T}_p$, and such that
\begin{align}\label{compatibleMetrics}
\delta(v,w)=\rho(\absolute{x-y}_p)
\end{align}
for any $x\in U_v$, $y\in U_w$, where $U_v,U_w\subset\mathds{Q}_p$ are the holed discs defined by the vertices $v,w$ of $T^\delta$ as a subtree of $\mathscr{T}_p$.
This necessitates the prime number $p$ to be sufficiently large.
\newline

The second step is to modify the Haar measure on $\mathds{Z}_p$ as follows: The compact open set $Z$ in $\mathds{Z}_p$ defined by $T_\delta$ is a holed disc. 
Set its measure to $1$, and divide for each node its measure equally among its child nodes.
 This yields a well-defined Radon measure on $Z$, and is denoted as $\nu$.

\begin{Lemma}\label{childrenEquality}
The Radon measure $\nu$ on $Z$ satisfies the property that
the child nodes of any
node of $T_\delta$ now
all have the same measure.
\end{Lemma}

\begin{proof}
This is immediate from the construction of $\nu$.
\end{proof}

\begin{Lemma}\label{differentialForm}
The measure $\nu$ is given as
\[
\nu(x)=
\phi(\absolute{f_Z(x)}_p)\,dx
\]
for some polynomial 
$f_Z\in\mathds{Q}_p[X]$ nowhere vanishing on $Z$, and a strictly increasing function $\phi\colon p^\mathds{Z}\to\mathds{R}_{>0}$. 
\end{Lemma}

\begin{proof} 
%We will show that the modification of the Haar measure to $\nu$ is obtained by a polynomial  function $f_Z\colon\mathds{Q}_p\to \mathds{Q}_p$, and then applying a function $\phi$ on its $p$-adic absolute values on $Z$.
%
%\smallskip
A polynomial function $f_Z$  as asserted can be obtained as follows:  
each leaf
node $v$ of $T_\delta$ corresponds to a unique  disc $B_v$ in $\mathds{Q}_p$ having a certain value $\nu(B_v)$ under the Radon measure $\nu$. Define  this value to correspond under the inverse map of $\phi$ to a unique element $p^{n_v}\in p^{\mathds{Z}}$ in a strictly increasing manner.
The function $f_Z$ is then defined like this: first, if necessary, cover the disc $B_v$ by disjoint discs whose radii are strictly smaller than $p^{n_v}$, and choose for each of these discs $B$ a reference element $a_B\in B$. 
Now find an interpolating $p$-adic polynomial $P\in\mathds{Q}_p[X]$ for these finitely many values at finitely many places.  
Then $f_Z=P$ has the desired property, because any  element $a_B+\epsilon$ of one of the small discs $B$ will be sufficiently close to $a_B$
such that
\[
\absolute{f_Z(a_B+\epsilon)}_p=\absolute{f_Z(a_B)}_p
\]
by construction.
%
%\smallskip
%The 
%Mahler expansion of $f_Z$ is
%\[
%g(x)=\sum\limits_{n\in\mathds{N}}a_n{x\choose n}
%\]
%with $a_n\in \mathds{Q}_p$. 
%The strategy is now to prove that $f_Z$ is continuous and can be interpolated by a polynomial $P\in\mathds{Q}_p[X]$. This then proves that $f_Z=P$ is actually polynomial. 
%\smallskip
%
%In order to see that $f_Z$ is continuous, it is sufficient to prove that $a_n$ converges to zero. 
%
%Furthermore, in order to see that $f_Z$ is also holomorphic, it is sufficient to prove that
%\begin{align}\label{AmiceCriterion}
%\frac{a_n}{n!}\to 0
%\end{align}
% \cite[Cor.\ III.10.2]{Amice1964}. 
%Mahler's Theorem says that
%\begin{align}\label{MahlerTheorem}
%a_n=\sum\limits_{k=0}^n(-1)^{n-k}{n\choose k}f_Z(k)
%=\Delta^n(f_Z)(0)
%\end{align}
%where 
%\[
%(\Delta f_Z)(x)=f_Z(x+1)-f_Z(x)
%\]
%is the difference operator
%\cite[Thm.\ 1]{Mahler1958}.
%The identity
%\begin{align}\label{PolynomialIdentity}
%\sum\limits_{k=0}^n
%(-1)^{n-k}{n\choose  k}%P(k)=0
%\end{align}
%for any polynomial $P(x)$ of degree less than $n$, 
%will thus allow to prove
%that $f_Z$ is continuous.
%Now, in order to be able to apply (\ref{PolynomialIdentity}),  observe that $f_Z$ has only finitely many values on  $\mathds{Q}_p$.  Hence, there exists a  polynomial $P\in\mathds{Q}_p[X]$ which interpolates the values of $f_Z$, and its degree 
%is at most $h$. Hence,  $f_Z=P$ is continuous and also a polynomial.
The  polyonomial $f_Z$ clearly does  not vanish on $Z$.
%
%\smallskip
%In order to see that this now proves the assertion, observe that 
%\[
%\nu(x)=\phi(\absolute{f(x)}_p)\,dx
%\]
%can now be established by  defining the  values for $\phi$ as
%\begin{align}\label{firstValues}
%\phi(\absolute{f_Z(x)}_p)=\frac{\nu(B_v)}{\mu_p(B_v)}
%\end{align}
%for $x\in B_v$ corresponding to a leaf node $v$ of $T_\delta$, in line with Lemma \ref{childrenEquality},
%where $\mu_p$ denotes the normalised Haar measure on $\mathds{Q}_p$, and any positive  value 
%\[
%\phi(\absolute{x}_p)
%\]
%for $x\in\mathds{Q}_p\setminus Z$ which is compatible with (\ref{firstValues}).
This now proves the assertion.
\end{proof}

\begin{remark}
Lemma \ref{differentialForm}
shows that a hierarchical dataset can be viewed as being sampled from a  
$1$-dimensional $p$-adic analytic manifold
having a regular algebraic differential $1$-form
\[
\omega=f_Z\,dx
\]
where now a notational conflict between the $p$-adic Haar measure $dx$ and the differential $1$-from $dx$ arises. The corresponding Radon measure $\nu$ can then be written as
\[
\nu(x)=\phi(\absolute{\omega(x)}_p)=\phi(\absolute{f_Z(x)}_p)\,\absolute{dx}_p
\]
written without that notational conflict by writing the Haar measure on $\mathds{Q}_p$ as $\absolute{dx}_p$ instead of $dx$, which from a conceptual point does make sense.
Since $f_Z$ is a polynomial, it follows that the differential $1$-form can be viewed as coming from a differential $1$-form on a Mumford curve defined over $\mathds{Q}_p$.
In other words, it is possible to view the finite dataset as sampled from an open compact subset of a  Mumford curve. An interesting question would be if and when different datasets could be viewed as samples from  open compact patches which form a covering of a Mumford curve. 
\end{remark}

%This can be answered with the following theorem:

%\begin{thm}[Pa\c{s}ol and Zaharescu, 2019]
%Fix a prime $p\ge3$. Let $f\colon \mathds{Z}_p\to\mathds{Z}_p$  be a function, and let 
%\[
%\tilde{f}(X)=\sum\limits_{n\in\mathds{N}}f(n)\frac{X^n}{n!}
%\]  
%If $e^{-X}\cdot\tilde{f}(X)$ converges on the disc $B(0,s)\subset\mathds{C}_p$ for some 
%$s>p^{-\frac{1}{p-1}}$, then $f$ is a locally analytic function of radius $(1+(p-1)\log_p(s))^{-1}$. In particular, $f$ is rigid analytic iff $e^{-X}\tilde{f}(X)$ is rigid analytic.
%\end{thm}

%\begin{proof}
%\cite[Thm.\ 3]{PZ2019}.
%\end{proof}

%\begin{remark}
%Since 
%\[
%e^{-X}\tilde{f}(X)=\sum\limits_{n\in\mathds{N}}\Delta^n(f)(0)\frac{X^n}{n!}
%=\sum\limits_{n\in\mathds{N}}\frac{a_n}{n!}X^n
%\]
%is in fact a polynomial,
%it follows that the Radon measure $\nu$ constructed above is the absolute value of a rigid-analytic holomorphic differential $1$-form on the Berkovich analytification $Z^{an}$ of $Z$. 
%\end{remark}

%%%%%%%%%%%%%%%%%%%
\subsection{Spectrum of $p$-adic operators from an ultrametric}

The goal is to study the operator
\[
\mathcal{H}_\delta u(x)
=\int_Z k_p(x,y)(u(y)-u(x))\,\nu(x)
\]
with kernel function $k_p(x,y)$ constructed from an ultrametric $\delta$ on the finite set $V$, and with
Radon measure $\nu$ as constructed in the previous subsection. The use of this measure instead of the restricted Haar measure will allow an explicit calculation also of the eigenvalues of the finite Laplacian matrix assciated with the operator.
\newline

Based on the observations made in \cite{Donoho1997} and \cite{Murtagh2007} that a tree representation of data naturally induces a multi-resolution analysis (MRA) and a Haar-like wavelet  basis, the authors of
\cite[\S3]{GNC2010}
construct such a Haar-like basis for any finite dataset represented by a rooted tree.
In the situation of this article, given 
an ultrametric $\delta$ on the vertices $V$ of a finite graph, this approach can be formulated 
as follows. Let
\[
B_r^\delta(v)=\mathset{
w\in V\mid \delta(v,w)\le r
}
\]
be a ball in $V$ of radius $r\ge0$. The nodes of the tree representation $T^\delta$ of $V$
given by the ultrametric $\delta$ are precisely such balls, or clusters, as they may also be called. These have no proper overlaps.
The tree $T^\delta$ is a tree with root $V$, and its leaf nodes are the elements of $V$.
The \emph{level} of a node $\mathfrak{n}$ of $T^\delta$ is 
the length of the geodesic connecting the root and $B$.
Define the set
\[
\mathcal{N}=\mathset{\text{nodes of $T^\delta$}}
\]
and
\[
\mathcal{N}_\ell:=
\mathset{\mathfrak{n}\in \mathcal{N}\mid \text{$\mathfrak{n}$ is at level $\ell$}}
\]
for $\ell\in\mathds{N}$.
Then obtain the vector spaces
\[
\mathds{C}(\mathcal{N}_\ell):=\mathset{f\colon V\to\mathds{C}\mid \forall \mathfrak{n}\in \mathcal{N}_\ell\colon f|_{\mathfrak{n}}=\text{const}}
\]
and a filtration
\[
0\subset \mathds{C}(\mathcal{N}_L)\subset\dots\subset \mathds{C}(\mathcal{N}_\ell)\subset\dots\subset\mathds{C}(\mathcal{N}_0)\cong\mathds{C}
\]
with 
\[
L=\max\mathset{\ell\in\mathds{N}\mid \mathcal{N}_\ell\neq\emptyset}
\]
the maximal level of a node in $T^\delta$.
Let
\[
W^\ell=\mathds{C}(\mathcal{N}_\ell)^\perp\subseteq \mathds{C}^V
\]
Then the function space $\mathds{C}^V$ has an orthogonal decomposition
\[
\mathds{C}^V=\bigoplus\limits_{\ell\in\mathds{N}}W^\ell\oplus \mathds{C}(\mathcal{N}_\ell) 
\]
as in \cite[(4)]{GNC2010}. The authors of \cite{GNC2010} construct a Haar-like orthonormal basis using that filtration. The aim here is to make this construction more explicit. 
For this, 
define
\[
\mathcal{C}(\mathfrak{n})=\mathset{\mathfrak{m}\in\mathcal{N}\mid\text{$\mathfrak{m}$ is a child node of $\mathfrak{n}$ in $T^\delta$}}
\]
and
\[
c(\mathfrak{n}):=\absolute{\mathcal{C}(\mathfrak{n})}
\]
where, in case $\mathfrak{n}$ is not a leaf node, $c(\mathfrak{n})>1$ holds true.
Then fix a bijection
\[
\mathcal{C}(\mathfrak{n})\to\mathds{Z}/c(\mathfrak{n})\mathds{Z}
\]
and thus transform the set $\mathcal{C}(\mathfrak{n})$ into a cyclic group.
Finally define
\[
\psi_{\mathfrak{n},\chi_{\mathfrak{n}}}(x)=
\nu(\mathfrak{n})^{-\frac12}\sum\limits_{\mathfrak{m}\in\mathcal{C}(\mathfrak{n})}\chi_{\mathfrak{n}}(\mathfrak{m})\,\Omega(x\in \mathfrak{m})
\]
where $x\in Z$, $\mathfrak{n}\in \mathcal{N}$, and 
\[
\chi_{\mathfrak{n}}\colon\mathcal{C}(\mathfrak{n})\to S^1\subset\mathds{C}
\]
is a non-trivial unitary character.
The set
\[
\mathcal{K}_{\delta}(V)=
\mathset{\psi_{\mathfrak{n},\chi_{\mathfrak{n}}}\mid \mathfrak{n}\in\mathcal{N},\;\chi_{\mathfrak{n}}\in\widehat{\mathcal{C}(\mathfrak{n})}\setminus \mathset{1}}
\]
where $\widehat{G}$ is the Pontryagin dual of a locally compact group $G$, 
is an explicit form of a Haar-like wavelet basis
of $\mathds{C}(\mathcal{N}_1)$.

\begin{definition}
The functions 
\[
\psi_{\mathfrak{n},\chi_{\mathfrak{n}}}(x)\in\mathcal{K}_\delta(V)
\]
are called \emph{ultrametric wavelets}.
\end{definition}

Notice that these ultrametric wavelets here coincide with  ultrametric wavelets defined in \cite{XK2005}. The only difference is that Khrennikov and Kozyrev relate them to a certain ultrametric on the space $Z$ such that the measure they use is given by the diameter of ultrametric discs. Here, the choice of a Radon measure is more flexible and independent of the choice of their ultrametric distance.
\newline

Together with $1_V$, the ultrametric wavelets yield an orthormormal basis of $\mathds{C}^V$.

\begin{Lemma}
It holds true that
\[
\int_Z\psi_{\mathfrak{n},\chi_{\mathfrak{n}}}(x)\,d\nu(x)=0
\]
for any $\psi_{\mathfrak{n},\chi_{\mathfrak{n}}}\in\mathcal{K}_\delta$.
\end{Lemma}

\begin{proof}
This is an immediate consequence of Lemma \ref{childrenEquality}, because that property forces each unit root in the integral to be equally weighted. Hence, their sum is zero.
\end{proof}

For vertices $v,w\in V$, write
\[
v\equiv w\mod \mathfrak{n}
\]
if they are both contained in the same child node of node $\mathfrak{n}$ of $T^\delta$. The 
child node of $\mathfrak{n}$ containing $v\in V$ will be denoted as
\[
[v]_{\mathfrak{n}}
\]
since $\equiv$ is an equivalence relation on the vertices contained in node $\mathfrak{n}$.
\newline

In the following assume a tree embedding of the tree $T^\delta$ into a $p$-adic Bruhat-Tits tree 
as in (\ref{treeEmbedding}).
Denote also the diameter of a set $A$ w.r.t.\ $\delta$ as
\[
\diam_\delta(A)=\sup\mathset{\delta(v,w)\mid v,w\in A}
\]
and the $p$-adic set distance as
\[
\dist_p(A,B)=\inf\mathset{\absolute{x-y}_p\mid x\in A,\;y\in B}
\]
for $A,B\subset\mathds{Q}_p$ compact subsets. Write also
$\dist_p(a,B)$ or $\dist_p(B,a)$ instead of $\dist_p(\mathset{a},B)$ or $\dist_p(B,\mathset{a})$.

\begin{thm}\label{SpectrumOfLaplacian2}
The space $L^2(Z,\nu)$ has an orthonormal basis consisting of 
the constant function $\nu(Z)^{-1}$, and 
the Kozyrev wavelets $\psi_{B,j}$ supported in $B\subset U_v$, where $v$ is  a vertex of $T^\delta$, and
the ultrametric wavelets $\psi_{\mathfrak{n},\chi_{\mathfrak{n}}}(x)$.
These different functions are all are eigenfunctions of $\mathcal{H}_\delta$. In the Kozyrev case, the corresponding eigenvalue is
\[
\lambda_B
=\phi(\absolute{a_v}_p)
\lambda_{B,v}-
\sum\limits_{v'\in V(T^\delta)\setminus\mathset{v}}
\phi(\absolute{a_{v'}}_p)\,
\delta(v,v')^{-\alpha}
\,\mu(U_{v'})
\]
where $\mu$ is the normalised Haar measure on $\mathds{Q}_p$,
\[
\nu(x)=\phi(\absolute{f_Z(x)}_p)\,dx
\]
is as in Lemma \ref{differentialForm},
\[
f_Z|_{U_v}=a_v\in\mathds{Q}_p
\]
and
\[
\lambda_{B,v}
=-\int_{U_v\setminus B}\dist_p(B,y)^{-\alpha}
\,dy-p^{d(1+\alpha)}
\]
whereas in the ultrametric wavelet case, it is 
\[
\gamma_{\mathfrak{n}}=
-\diam_\delta(\mathcal{C}(\mathfrak{n}))^{-\alpha}
\absolute{\mathcal{C}(\mathfrak{n})}\nu(\mathfrak{n}\setminus\mathfrak{m}(x))
\]
where $\mathfrak{n}$ is a node of $T^\delta$, and $\mathfrak{m}(x)$ is the child node of $\mathfrak{n}$ containing $x$. The quantity $\gamma_{\mathfrak{n}}$ does not depend on the particular child node $\mathfrak{m}(x)$. The operator  $\mathcal{H}_\delta$ is self-adjoint and negative definite on $L^2(Z,\nu)$.
\end{thm}

\begin{proof}
The tree embedding (\ref{treeEmbedding}) identifies the vertices $v$ of $T^\delta$ with compact open subsets $U_v$ of $\mathds{Q}_p$ which are discs with finitely many  sub-discs removed. 
These satisfy the property (\ref{compatibleMetrics}).
Even more, it holds true that
\[
\delta(v,w)=\rho(\absolute{x-y}_p)
=\rho(\dist_p(U_v,U_w))
\]
which applies to 
a Kozyrev wavelet as
\begin{align*}
\mathcal{H}_\delta
\psi_{B,j}(x)
&=
\int_Z k_p(x,y)\left(
\psi_{B,j}(y)-
\psi_{B,j}(x)
\right)
\nu(x)
\\
&=\int_Z k_p(x,y)
\phi(\absolute{f_Z(y)}_p)
\left(
\psi_{B,j}(y)-\psi_{B,j}(x)
\right)
\,dy
\\
&=\sum\limits_{v\in V(T^\delta)}
\phi(\absolute{a_v}_p)
\int_{U_v} k_p(x,y)
\left(
\psi_{B,j}(y)-\psi_{B,j}(x)
\right)
\end{align*}
where the second equality comes from Lemma \ref{differentialForm}, and the third uses the notation from the beginning of Section \ref{RadonMeasure} with $a_v\in\mathds{Q}_p$, making use of the fact that $\absolute{f_Z}_p$ is constant on each set $U_v$ with $v\in V(T^\delta)$.

\smallskip
Assume that  $B\subseteq U_v$ for vertex $v\in T^\delta$.
Then
\begin{align*}
\mathcal{H}_\delta\psi_{B,j}(x)
&=\phi(\absolute{a_v}_p)
\int_B k_p(x,y)\left(
\psi_{B,j}(y)-
\psi_{B,j}(x)
\right)\,dy
\\
&+
\sum\limits_{v'\in V(T^\delta)\setminus\mathset{v}}
\phi(\absolute{a_{v'}})\,
\delta(v,v')^{-\alpha}
\cdot
\int_{U_{v'}}
\left(\psi_{B,j}(y)-\psi_{B,j}(x)\right)\,dy
\end{align*}
The kernel function
\[
T_v(x,y):=k_p(x,y)\,\Omega(x\in U_v)\times \Omega(y\in U_v)
\]
satisfies the condition 
\[
\absolute{x-y}_p=\text{const.}
\quad
\Rightarrow
\quad
T_v(x,y)=\text{const.}
\]
Hence, according to \cite[Thm.\ 3]{Kozyrev2004}, the $p$-adic Kozyrev wavelets supported in the sets $U_v$ are eigenfunctions of the operator
\[
\mathcal{T}_v u(x)
=\int_{U_v}T_v(x,y)(u(y)-u(x))\,dy
\]
with corresponding eigenvalue
\[
\lambda_{B,v}
=-\int_{U_v\setminus B}
\absolute{p^{-d}n-y}_p^{-\alpha}\,dy
-p^{d(1+\alpha)}
=
-\int_{U_v\setminus B}\dist(B,y)^{-\alpha}
\,dy-p^{d(1+\alpha)}
\]
where 
\[
B=\mathset{z\in\mathds{Q}_p\mid
\absolute{ p^{-d}n-z}_p\le 1}
\]
is the support of the wavelet for some $n\in\mathds{Q}_p$ being the common fractional part of its residue class in $\mathds{Q}_p/\mathds{Z}_p$. It follows that the eigenvalue under the operator $\mathcal{H}_\delta$ is
\[
\lambda_B
=\phi(\absolute{a_v}_p)
\lambda_{B,v}-
\sum\limits_{v'\in V(T^\delta)\setminus\mathset{v}}
\phi(\absolute{a_{v'}}_p)\,
\delta(v,v')^{-\alpha}
\,\mu(U_{v'})
\]
for the Haar measure $\mu$ on $\mathds{Q}_p$. This proves the Kozyrev wavelet case.

\smallskip
Now, let
$\psi_{\mathfrak{n}},\chi_{\mathfrak{n}}$ be a Haar-like wavelet. It is constant on the compact open subsets of $Z$ defined by the children of $\mathfrak{n}$.
Their number equals the number of children per vertex summed up. This equals the number of leaf nodes of $T^\delta$ which equals $\absolute{V}$, the dimension of the subspace in $L^2(Z,\nu)$ spanned by the functions constant on the discs $U_v$. Each $\psi_{\mathfrak{n}}$ has norm $1$, and any two distinct Haar-like wavelets is are orthogonal, because either their supports are disjoint, or one is strictly contained in the other and the corresponding wavelet is constant on the support of the other. Hence, the inner product calculates the sum of the corresponding $n$-th roots which vanishes.
Hence, 
\begin{align*}
\mathcal{H}_\delta
\psi_{\mathfrak{n},\chi_{\mathfrak{n}}}(x)
&=-\sum\limits_{\mathfrak{m}\in \mathcal{C}(\mathfrak{n})\atop \mathfrak{m}\neq \mathfrak{m}(x)}
\delta(\mathfrak{m}(x),\mathfrak{m})^{-\alpha}
\,\nu(\mathfrak{m})
\left(\psi_{\mathfrak{n},\chi_{\mathfrak{n}}}(y)
-\psi_{\mathfrak{n},\chi_{\mathfrak{n}}}(x)\right)
\\
&=-\diam_{\delta}(\mathcal{C}(\mathfrak{n}))^{-\alpha}\absolute{\mathcal{C}(\mathfrak{n})}
\nu(\mathfrak{n}\setminus\mathfrak{m}(x))\,
\psi_{\mathfrak{n},\chi_{\mathfrak{n}}}(x)
\end{align*}
because for any primitive $N$-th root of unity $\zeta$, it  holds true that
\[
\sum\limits_{\ell=0\atop\ell\neq j}^{N-1}
\zeta^j-\zeta^\ell
=\zeta^j
\sum\limits_{j=0}^{N-1}
1-\zeta^j
=N\zeta^j 
\]
for $j\in\mathset{0,\dots,N-1}$. This proves the Haar-like wavelet case.
The eigenvalue $\gamma_{\mathfrak{n}}$ does not depend on the child node $\mathfrak{m}(x)$ by construction of measure $\nu$.
Self-adjointness of $\mathcal{H}_\delta$ is now follows from the fact that these various wavelets form a basis of $L^2(Z,\nu)$ together with the constant function, and by symmetry of  $\mathcal{H}_\delta$. Self-adjointness and negative semi-definiteness of $\mathcal{H}_\delta$ are now immediate.
\end{proof}

\begin{remark}
Observe that Theorem \ref{SpectrumOfLaplacian2}  is a special case of \cite[Thm.\ 10]{XK2005}. The only difference is that the spectrum is explicitly calculated in terms of the presentation of the Radon measure $\nu$ as being constructed from the differential form $\omega$ and the function 
$\phi$.
\end{remark}

%%%%%%%%%%%%%%%%%%%5
\section{Tree truncation: an approximation method}

    Now we present an approximation method which allow us to substitute the original tree with a simpler one, concretely, a truncation of $T^{\delta}$.  First, notice that if each ball attached to the domain $Z$ has the same radii, then the measures $dx$ and $\nu$ coincide. Without losing generality, we assume that all the $p$-adic balls generating $Z$ have the same radii. 
\newline
    
    The kernel $k(x,y)$ is constructed by assigning a transition rate to every two points $x,y$ in an extended $p$-adic domain $Z$ constructed by attaching to every leaf of $T^{\delta}$ a $p$-adic ball in such a way that the ultrametric distance $\delta$ coincides with the value of a radial function $\rho(|x-y|_p)$. The kernel with respect to the ultrametric structure of the tree for points in different vertices $x\in U_{w}$, $y\in U_{v}$, for $w\neq v$, and introducing a transition rate by the Vladimirov kernel $|x-y|_p^{-\alpha}$ for $x,y\in U_{v}$. For every level $\ell$ we can generate a subtree by cutting $T^{\delta}$ to this respective level, that is, the new leaves of this tree are given by the nodes $\mathfrak{n}\in \mathcal{N}_{\ell}$, cf. the previous section. By repeating the process, we construct a new $p$-adic domain given by 
    \[
    Z_{\ell}
    =\bigcup_{\mathfrak{n}
    \in \mathcal{N}_{\ell}} U_{\mathfrak{n}},
    \]
where it is clear that $Z\subseteq Z_{\ell}$ for all levels. 
We then define the kernel function $k_{\bullet}^{\ell}(x,y)$ in the same way as $k_{\bullet}(x,y)$. These functions differ in the following case: Let $x\in U_{w}$, $y\in U_{v}$, for $w\neq v$, be such that $x,y\in U_{\mathfrak{n}}$. Then
\[
k_{\bullet}^{\ell}(x,y)=|x-y|_p^{-\alpha}
\]
for $\ell>0$.
\newline

Let
\[
\mathcal{H}_{\bullet}^{\ell} u(x)
=\int_{Z_{\ell}}k_{\bullet}^{\ell}(x,y)(u(y)-u(x))\,dy,
\]
for $\ell\in\mathds{N}$.
%\begin{definition}
The operator $\mathcal{H}_\bullet^\ell$ 
%is called the \emph{hierarchical Parisi matrix} associated with $\bullet$ and 
acts on $L^2(Z_\ell,d_\ell x)$ as well as on $L^2(Z_\ell,\nu_\ell)$.
%\end{definition}

\begin{thm}\label{errorSemigroups}
    Let $T_{\bullet}^{\ell}(t)$ and $T_{\bullet}(t)$ be the semigroups attached to the operators $\mathcal{H}_{\bullet}^{\ell}$ and $\mathcal{H}_{\bullet}$, respectively.  Let $u\in C(Z)$ and $t\in [0,\gamma]$, then
\begin{equation*}
    \begin{split}
        ||T_{\bullet}^{\ell}(t)u -T_\bullet (t)u||_{\infty}\leq  \gamma||u||_{\infty}\left(\sum_{u\neq v, [u]_{\mathfrak{n}}\equiv[v]_{\mathfrak{n}}}C_{u,v}\Vol(U_v)+\Vol(Z_{\ell}\setminus Z)\max_{x\in Z, y\in Z_{\ell}\setminus Z} |k_{\bullet}^{\ell}(x,y)|\right).
    \end{split}
\end{equation*}
\end{thm} 
\begin{proof}
    We have

\begin{align*}
\absolute{\mathcal{H}_\bullet^{\ell} u(x)-\mathcal{H}_\bullet u(x)}
&\leq \int_{Z} |k_{\bullet}^{\ell}(x,y)-k_{\bullet}(x,y)||u(y)-u(x)|\,dy 
\\
&+ \int_{Z_{\ell}\setminus Z} |k_{\bullet}^{\ell}(x,y)||u(y)-u(x)|\,dy,
\end{align*}
The first term can be expressed in terms of the discrepancy of the kernel functions in the following way 
\begin{align*}
\int_{Z}& |k_{\bullet}^{\ell}(x,y)-k_{\bullet}(x,y)|
|u(y)-u(x)|\,dy 
\\
&=\sum_{w\neq v\atop [w]_{\mathfrak{n}}\equiv[v]_{\mathfrak{n}}} \Omega(x\in U_{w})\int_{U_{v}} \left||x-y|^{-\alpha} -k(x,y)^{-\alpha}\right||u(y)-u(x)|\,dy ,
\end{align*}
where the sum runs over $\mathfrak{n}\in \mathcal{N}_{\ell}$. 

\smallskip
By the mean value theorem,  it can be shown that 
\[\left||x-y|^{-\alpha} -k(x,y)^{-\alpha}\right|\leq \alpha\frac{\left||x-y|_p -k(x,y)\right|}{\min(|x-y|_p,k(x,y))^{\alpha+1}}\]
Each integral in the sum corresponds to the case $x\in U_{w}$ and $y\in U_v$, for $w\neq v$, therefore 
\[\alpha\frac{\left||x-y|_p -k(x,y)\right|}{\min(|x-y|_p,\delta(x,y))^{\alpha+1}}=\alpha\frac{\left|\dist_p(U_w,U_v) -k(U_w,U_v)\right|}{\min(\dist_p(U_w,U_v),k(U_w,U_v))^{\alpha+1}}=:C_{w,v}\]
Then
\begin{align*}
\Omega(x\in U_{w})&\int_{U_{v}} \left||x-y|^{-\alpha} -k(x,y)^{-\alpha}\right||u(y)-u(x)|\,dy
\\
&\leq C_{w,v}\,\Omega(x\in U_{w})\left(||u||_{\infty}\Vol(U_v)^{1/2}+|u(x)|\Vol(U_v)\right),
\end{align*}
and
\[\left|\left|\int_{Z} |k^{\ell}_{\bullet}(x,y)-k_{\bullet}(x,y)||u(y)-u(x)|dy\right|\right|_{\infty}\leq 2||u||_{\infty}\sum_{w\neq v\atop [w]_{\mathfrak{n}}\equiv[v]_{\mathfrak{n}}}C_{w,v}\Vol(U_v).\]
The following estimate follows
\begin{align*}
\|\mathcal{H}_{\bullet}^{\ell} u(x)-\mathcal{H}_{\bullet} u(x)||_{\infty}
&\leq 2||u||_{\infty}\sum_{w\neq v\atop [w]_{\mathfrak{n}}\equiv[v]_{\mathfrak{n}}}C_{w,v}\Vol(U_v)
\\
&+||u||_{\infty}\Vol(Z_{\ell}\setminus Z)\max_{x\in Z, y\in Z_{\ell}\setminus Z} |k_{\bullet}^{\ell}(x,y)|
\end{align*}
By Proposition $3.3$ we have that $\mathcal{H}_{\bullet}$ (and similarly $\mathcal{H}_{\bullet}^{\ell}$) is the infinitesimal generator of a Feller semigroup, and therefore a contraction semigroup. Let $t\geq0$ , then, we have the following estimate of the corresponding evolution processes 
\begin{equation*}
    \begin{split}
        ||T_{\bullet}^{\ell}(t)u -T_{\bullet} (t)u||_{\infty}&\leq \int_{0}^{t}||T(s)_{\delta}^{\ell}u||_{\infty}|| (\mathcal{H}_{d_E}-\mathcal{H}_\delta)u||_{\infty}||T_{\delta}(t-s)u||_{\infty}ds \\
        &\leq t || \mathcal{H}_{\delta}^{\ell}u-\mathcal{H}_\delta u ||_{\infty} \\
        &\leq t\left(2\sum_{u\neq v, [u]_{\mathfrak{n}}\equiv[v]_{\mathfrak{n}}}C_{u,v}\Vol(U_v)+\Vol(Z_{\ell}\setminus Z)\max_{x\in Z, y\in Z_{\ell}\setminus Z} |k_{\bullet}^{\ell}(x,y)|\right).
    \end{split}
\end{equation*}

\end{proof}

Therefore, we see that the similarity of the processes on a finite interval of time $t\in[0,T]$ , will depend in various factors, one is of course the cut off level $\ell$, which is quantified on the contribution of $\Vol(Z_{\ell}\setminus Z)$. The second factor are the coefficients $C_{u,v}$ which depend on how similar both ultrametric are. Finally, the first sumand depend on $\Vol(U_{v})$, therefore, this value will be smaller as the number of leafs on the original tree $T^{\delta}$ increases. \newline 

\begin{thm}\label{errorSemigroups2}
    Let $\mathcal{H}_a$ and $\mathcal{H}_b$, where $a,b\in\mathset{\kappa,d_E,\delta}$. Let $T_a$ and $T_b$  the respective attached subgroups. Then the following holds 
    \[||T_{a}(t) -T_b||_{\infty} \leq 2t \left(\sum_{u\neq v, [u]_{\mathfrak{n}}\equiv[v]_{\mathfrak{n}}}\tilde{C}_{u,v}\Vol(U_v)\right).\]
    
\end{thm}
\begin{proof}
    In a similar way we can obtain an estimate for the operator $\mathcal{H}_{d_{E}}$, where the shortest-path distance $d_E$ is one of the possible instances of $\bullet$ as in (\ref{possibleBullets}). 
In this case, by defining
\[
\tilde{C}_{w,v}=\alpha\frac{\left|\dist_E(U_w,U_v) -\delta(U_w,U_v)\right|}{\min(\dist_E(U_w,U_v),\delta(U_w,U_v))^{\alpha+1}}=\alpha\frac{\left|\dist_E(U_w,U_v) -\delta(U_w,U_v)\right|}{\delta(U_w,U_v)^{\alpha+1}}
\]
for $w\neq v$, where $\dist_E$ is given as
\[
\dist_E(A,B)=\min\mathset{d_E(a,b)\mid a\in A,\;b\in B}
\]
with $A,B\subseteq Z$,
and the last equality holds true since $\delta(x,y)\leq d_{E}(x,y)$, then we obtain 
\[||\mathcal{H}_{d_E} u(x)-\mathcal{H}_\delta u(x)||_{L^2}\leq 2||u||_{L^2}\left(\sum_{u\neq v, [u]_{\mathfrak{n}}\equiv[v]_{\mathfrak{n}}}\tilde{C}_{u,v}\Vol(U_v)\right).\]
We obtain a similar bound for the corresponding evolution processes. 
\[||T_{d_{E}}(t) -T_\delta (t)||_{\infty} \leq 2t \left(\sum_{u\neq v, [u]_{\mathfrak{n}}\equiv[v]_{\mathfrak{n}}}\tilde{C}_{u,v}\Vol(U_v)\right).\]

\end{proof}

Let
\[
k_v(x,y)=\absolute{x-y}_p^{-\alpha_v}
\]
for $x,y\in v\subset\mathds{Z}_p$ and $\alpha_v>0$. The following approximation is used:
\[
v_n:=\mathset{[x]_n\mid x\in v}\subset\mathds{Z}_p/p^n\mathds{Z}_p
\]
for $n\in\mathds{N}$ sufficiently large such that
$p^{-n}$ is smaller than the $p$-adic diametre of the disc $v$. Set
\[
k_{v_n}([x]_n,[y]_n):=
\begin{cases}
\absolute{[x]_n-[y]_n}^{-\alpha_v},&
[x]_n\neq[y]_n
\\
0,&\text{otherwise}
\end{cases}
\]
as the kernel function of the operator
\[
\mathcal{H}_v^{(n)}u_n([x]_n)
=\int_{v_n} k_{v_n}([x]_n,[y]_n)(u_n([y]_n)-u_n([x]_n))\,d\nu_n([y]_n)
\]
with the induced measure $\nu_n$ on $v_n$ coming from $\nu$ on $v$, and for functions $u_n$ on $v_n$, forming the finite-dimensional vector space $X_n$. Further. set $X=C(v,\norm{\cdot}_\infty)$, and define the embedding and projection operators
\begin{align*}
E_n u_n(x)&=u_n([x]_n),\quad u_n\in X_n
\\
P_n u([x]_n)&=u(x_n),\quad u\in X
\end{align*}
where $x_n$ is a fixed representative of the class $[x]_n$ for every $x\in v$.

\begin{thm}\label{CPerror}
Let $u\in X\times[0,\infty)$ be a solution of the Cauchy problem for the heat equation
\begin{align*}
\frac{\partial}{\partial t}u(x,t)&-\mathcal{H}_{\bullet}u(x,t)=0
\\
u(x,0)&=u_0(x)\in X
\end{align*}
and $u_n\in X_n\times[0,\infty)$ a solution of that for
\begin{align*}
\frac{\partial}{\partial t}u_n([x]_n,t)&-\mathcal{H}_{\bullet}^{(n)} u_n([x]_n,t)=0
\\ u_n([x]_n,0)&=P_nu_0([x]_n)\in X_n
\end{align*}
Then it holds true that
\[
\lim\limits_{n\to\infty}\sup\limits_{0\le t\le\tau}\norm{E_nu_n([x]_n,t)-u(x,t)}=0
\]
for $\tau\ge 0$.
\end{thm}

\begin{proof}
Verify that the criteria (A), (B), (C') of
\cite[Ch.\ 5.4]{Miklavcic2001}
hold true. Then apply
\cite[Thm.\ 5.4.7]{Miklavcic2001} in order to approximate the solutions in the asserted manner.
\end{proof}
The solution $u(x,t)$ can be expressed in terms of the eigenvectors of the operator $\mathcal{H}_{\bullet}$ and its eigenvalues in the following way 
\[u(x,t)=\sum_{I} e^{t\lambda_{I}}\phi_I(x)+\sum _{B,j}e^{t\lambda_{B,j}}\psi_{B,j}(x)\]
where $\phi_I$ correspond to the eigenvectors attached to the Laplacian representation on the repective finite dimensional space, and $\psi_{B},j$  are the Kozyrev eigenvalues supported in the ball $B\subset Z$. In a similar way, we can express the function $E_nu_n([x]_n,t)$ in terms of a finite expansion of such functions: 
\[E_nu_n([x]_n,t)=\sum_{I} e^{t\lambda_{I}}\phi_I(x)+\sum _{B,j, \Vol(B)>l}e^{t\lambda_{B,j}}\psi_{B,j}(x),\]
therefore, we have the following estimate 
\[\sup\limits_{0\le t\le\tau}\norm{E_nu_n([x]_n,t)-u(x,t)}=\sup\limits_{0\le t\le\tau}\norm{\sum _{B,j, \Vol(B)\leq l}e^{t\lambda_{B,j}}\psi_{B,j}(x)}\]

%%%%%%%%%%%%%%%%55
%\section{Discussion}

%\ToDo{

%- Is this section needed?
%}

%%%%%%%%%%%%%%%%%%%
\section*{Acknowledgements}

%\'Angel Mor\'an Ledezma and 
Andrew and John Bradley, Martin Breunig, Paulina Halwas, Markus Jahn, Fionn Murtagh, Leon Nitsche, David Weisbart, and Wilson Z\'u\~niga-Galindo are warmly thanked for fruitful discussions. 
David Weisbart and the University of California Riverside are thanked for hosting the first named author, where partial results of this article were produced and presented. 
This work is supported by the Deutsche Forschungsgemeinschaft under project number 469999674.

%%%%%%%%%%%%%%%%
\bibliographystyle{plain}
\bibliography{biblio}

\begin{thebibliography}{10}

\bibitem{Alexandrov1937}
P.S. Alexandrov.
\newblock Diskrete {R\"aume}.
\newblock {\em Matematicheskii Sbornik (N.S.)}, 2:501--518, 1937.

\bibitem{Angulo2019}
J.~Angulo.
\newblock Hierarchical {Laplacian} and its spectrum in ultrametric image
  processing.
\newblock In B.~Burgeth et~al., editor, {\em ISMM 2019}, LNCS 11564, pages
  29--40, Switzerland AG, 2019. Springer Nature.

\bibitem{brad_HeatMumf}
P.E. Bradley.
\newblock Heat equations and wavelets on {Mumford} curves and their finite
  quotients.
\newblock {\em Journal of Fourier Analysis and Applications}, 29:62, 2023.

\bibitem{brad_thetaDiffusionTateCurve}
P.E. Bradley.
\newblock Theta-induced diffusion on {Tate} elliptic curves over
  non-archimedean local fields.
\newblock arXiv:2312.03570 [math.NT], 2023.

\bibitem{HearingGenusMumf}
P.E. Bradley.
\newblock Heat equations and hearing the genus on $p$-adic {Mumford} curves via
  automorphic forms.
\newblock arXiv:2402.02869 [math.NT], 2024.

\bibitem{SchottkyInvariantDiffusion}
P.E. Bradley.
\newblock Schottky-invariant $p$-adic diffusion operators.
\newblock arXiv:2405.17586 [math.AG], 2024.

\bibitem{ScaleHilbert}
P.E. Bradley and M.W. Jahn.
\newblock On the behaviour of $p$-adic scaled space filling curve indices for
  high-dimensional data.
\newblock {\em The Computer Journal}, 65(2):310--330, 2022.

\bibitem{BL_shapes_p}
P.E. Bradley and \'A. {Mor\'an Ledezma}.
\newblock Hearing shapes with $p$-adic {Laplacians}.
\newblock {\em J. Math. Phys.}, 64:113502, 2023.

\bibitem{Dijkstra1959}
E.W. Dijkstra.
\newblock A note on two problems in connexion with graphs.
\newblock {\em Numerische Mathematik}, 1:269--271, 1959.

\bibitem{Donoho1997}
D.~Donoho.
\newblock {CART} and best-ortho-basis: a connection.
\newblock {\em The Annals of Statistics}, 25(5):1870--1911, 1997.

\bibitem{GNC2010}
M.~Gavish, B.~Nadler, and R.R. Coifman.
\newblock Multiscale wavelets on trees, graphs and high dimensional data:
  Theory and applications to semi supervised learning.
\newblock In {\em Proceedings of the 27th International Conference on Machine
  Learning}, Haifa, Israel, 2010.

\bibitem{GL2023}
E.~Gorman and M.E. Lladser.
\newblock Sparsification of large ultrametric matrices: insights into the
  microbial tree of life.
\newblock {\em Proc. R.Soc.}, A 479:20220847, 2023.

\bibitem{Kahn1962}
A.B. Kahn.
\newblock Topological sorting of large networks.
\newblock {\em Communications of the ACM}, 5(11)):558--562, 1962.

\bibitem{XK2005}
A.Yu. Khrennikov and S.V. Kozyrev.
\newblock Wavelets on ultrametric spaces.
\newblock {\em Appl. Comput. Harmon. Anal.}, 19:61--76, 2005.

\bibitem{Kochubei2018}
A.N. Kochubei.
\newblock Linear and nonlinear heat equations on a $p$-adic ball.
\newblock {\em Ukr. Math. J.}, 70:217--231, 2018.

\bibitem{Kozyrev2002}
S.~V. Kozyrev.
\newblock Wavelet theory as $p$-adic spectral analysis.
\newblock {\em Izvestiya: Mathematics}, 66(2):367--376, 2002.

\bibitem{Kozyrev2004}
S.V. Kozyrev.
\newblock $p$-adic pseudodifferential operators and $p$-adic wavelets.
\newblock {\em Theoretical and Mathematical Physics}, 138(3):322--332, 2004.

\bibitem{Miklavcic2001}
M.~Miklav\v{c}i\v{c}.
\newblock {\em Applied Functional Analysis and Partial Differential Equations}.
\newblock World Scientific, Singapore, reprinted edition, 2001.

\bibitem{Murtagh2007}
F.~Murtagh.
\newblock The {Haar} wavelet transform of a dendrogram.
\newblock {\em Journal of Classification}, 24:3--32, 2007.

\bibitem{PW2024}
T.~Pierce and D.~Weisbart.
\newblock Brownian motion in the $p$-adic integers is a limit of discrete time
  random walks.
\newblock arXiv:2407.05561 [math.PR], 2024.

\bibitem{RTV1986}
R.~Rammal, G.~Toulouse, and M.A. Virasoro.
\newblock Ultrametricity for physicists.
\newblock {\em Rev. Mod. Phys.}, 58:765, 1986.

\bibitem{Taibleson1975}
M.H. Taibleson.
\newblock {\em Fourier analysis on local fields}.
\newblock Princeton University Press, Princeton, N.J., University of Tokyo
  Press, Tokyo, 1975.

\bibitem{VVZ1994}
V.S. Vladimirov, I.V. Volovich, and E.I. Zelenov.
\newblock {\em $p$-adic Analysis and Mathematical Physics}.
\newblock Series on Soviet and East European Mathematics, 1. World Scientific
  Publishing Co., Inc., River Edge, NJ, 1994.

\bibitem{ZunigaNetworks}
W.~{Z\'{u}\~{n}iga-Galindo}.
\newblock Reaction-diffusion equations on complex networks and {Turing}
  patterns, via $p$-adic analysis.
\newblock {\em Journal of Mathematical Analysis and Applications},
  491(1):124239, 2020.

\end{thebibliography}

\end{document}